\documentclass[10pt,conference]{IEEEtran}

\usepackage[letterpaper, bottom=0.95in, top=0.75in, left=0.625in, right=0.625in]{geometry}

\usepackage{amsmath,graphicx}
\usepackage{bbm}
\usepackage{amsmath, amssymb, amscd, amsthm, amsfonts}
\usepackage{subfiles}
\usepackage{color}

\usepackage[caption=false,font=footnotesize]{subfig}

\usepackage[backend=biber, style=ieee, dashed=false,isbn=false,doi=false,url = false,bibencoding=utf8,citestyle=numeric-comp]{biblatex}
\usepackage[linesnumbered,ruled,vlined]{algorithm2e}
\usepackage[noend]{algpseudocode}

\SetKwInput{KwInput}{Input}                
\SetKwInput{KwOutput}{Output}              

\SetCommentSty{mycommfont}

\newtheorem{proposition}{Proposition}

\title{Timing Recovery and Sequence Detection for Integrate-and-Fire Time Encoding Receivers}

\author{
\IEEEauthorblockA{Neil Irwin Bernardo}\\
\IEEEauthorblockA{\textit{Electrical and Electronics Engineering Institute} \\
\textit{University of the Philippines Diliman}\\
Quezon City, Philippines\\
neil.bernardo@eee.upd.edu.ph}
}

\begin{document}
\maketitle
\begin{abstract}
Recent advances in neuromorphic signal processing have introduced time encoding machines as a promising alternative to conventional uniform sampling for low-power communication receivers. In this paradigm, analog signals are converted into event timings by an integrate-and-fire circuit, allowing information to be represented through spike times rather than amplitude samples. While event-driven sampling eliminates the need for a fixed-rate clock, receivers equipped with integrate-and-fire time encoding machines, called \emph{time encoding receivers}, often assume perfect symbol synchronization, leaving the problem of symbol timing recovery unresolved. This paper presents a joint timing recovery and data detection framework for integrate-and-fire time encoding receivers. The log-likelihood function is derived to capture the dependence between firing times, symbol timing offset, and transmitted sequence, leading to a maximum likelihood formulation for joint timing estimation and sequence detection. A practical two-stage receiver is developed, consisting of a timing recovery algorithm followed by a zero-forcing detector. Simulation results demonstrate accurate symbol timing offset estimation and improved symbol error rate performance compared to existing time encoding receivers.
\end{abstract}
\section{Introduction}
\label{section:intro}

The time encoding machine (TEM) is a bio-inspired signal representation framework that converts an analog waveform into a sequence of timing events \cite{Tuckwell_1988}. Unlike conventional sampling methods, where the signal is sampled at uniform intervals, the TEM encodes information in the time domain through the event times at which an internal state, e.g. accumulated voltage, crosses a specified threshold \cite{miskowicz2015event}. Such event-driven sampling eliminates the need for a fixed-rate clock during signal acquisition, allowing for asynchronous operation and reduced power consumption. Among various implementations of TEM, the integrate-and-fire TEM (IF-TEM) has received particular attention due to its simplicity and close correspondence to the neuronal firing process \cite{NeuronalDynamics:2014}.

Theoretical works have established conditions for perfect recovery of band-limited signals from IF firing times and demonstrated the stability of such encodings under bounded perturbations \cite{Lazar:2004}. Recovery algorithms for IF-TEM have also been proposed to handle finite rate of innovation (FRI) signals, signals residing in shift-invariant (SI) spaces, and sparse signals \cite{Gontier:2014,Alexandru:2021,Kamath:2022,Naaman:2022,Florescu:2023,Kamath:2023}. In addition to these theoretical developments, practical implementations of TEM have been investigated in several applications, including neuromorphic hardware, power-efficient data converters, image and video processing, and biomedical systems, where low-power operation and asynchronous event-driven processing are advantageous \cite{Lazar:2011,Kong:2012,Lazar:2014,Gutierrez-Galan:2022,Fu:2024,Naaman:2024,Naaman:2024b,Arora:2025,Mulleti:2025}.

Recent works have investigated the use of time encoding and neuromorphic signal processing principles in communication system design \cite{Katsaros:2025,Biswas:2026,Bernardo2025SIPS}. In \cite{Katsaros:2025}, a fully-neuromorphic receiver based on leaky integrate-and-fire neurons was shown to achieve competitive bit error rate performance with power consumption in the microwatt range. However, this neuromorphic receiver remains restricted to repetition-coded binary signaling and rectangular pulse shapes, with no demonstrated extensions to higher-order modulation formats or more general pulse designs. Building on this concept, another study \cite{Bernardo2025SIPS} introduced a time encoding receiver that supports multi-level signaling and Gaussian pulse shaping filters. The inter-symbol interference (ISI) arising from the pulse shaping filter is mitigated by using a zero-forcing (ZF) equalizer at the receiver. Similar to the neuromorphic receiver in \cite{Katsaros:2025}, the time encoding receiver in \cite{Bernardo2025SIPS} estimates the transmitted symbols from the number of firing events per symbol period.

 \begin{figure*}[t!]
    \centering
    \includegraphics[scale = .9]{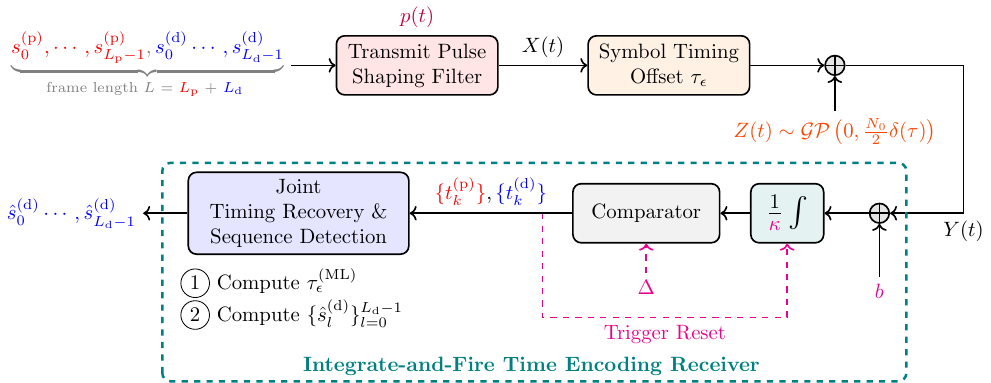}
    \caption{System diagram of the proposed integrate-and-fire time encoding receiver.}
    \label{fig:IFTEM_sysmodel}
\end{figure*}

Despite these advances, the problem of timing recovery in time-encoding receivers has not been explicitly addressed in the literature. Existing approaches typically assume perfect synchronization, i.e., that the receiver has prior knowledge of the symbol timing or that symbol boundaries are pre-aligned. However, in practice, even a small timing offset can disrupt the correspondence between transmitted symbols and observed firing times. For instance, an incorrectly aligned observation window in \cite{Katsaros:2025,Bernardo2025SIPS} may count spikes belonging to adjacent symbols, leading to erroneous spike counts and degraded detection performance.  The nonlinear and asynchronous nature of the mapping from the input waveform to the spike events further complicates this task, rendering conventional timing recovery techniques ineffective. This gap motivates the development of a framework for estimating the timing information and the transmitted symbols directly from the observed firing times.

In this paper, we present a timing recovery and sequence detection scheme for integrate-and-fire (IF) time encoding receivers. Specifically, we develop a statistical framework that characterizes the relationship among firing times, timing offset, and transmitted symbols, and formulate a maximum likelihood (ML) estimation problem for jointly recovering the timing offset and data sequence. To achieve computational tractability, we adopt a two-stage approach in which the timing offset is first estimated from the firing times induced by the pilot sequence, followed by data sequence detection using the firing times induced by the data symbols and the timing information estimated in the previous stage. The main contributions of this work are as follows:
\begin{itemize}
\item We derive the log-likelihood function for the joint estimation of the symbol timing and data sequence given the observed firing times and known pilot symbols. This formulation enables a unified statistical treatment of timing recovery and data detection.
\item We identify a necessary condition for the ML estimate of the timing offset and exploit this necessary condition to design a practical timing recovery algorithm.
\item We propose a sequence detection method based on a ZF receiver that operates directly on the firing times instead of the firing count per symbol period. The proposed detector effectively mitigates ISI due to the pulse shaping filter and can operate at a low firing rate.
\item We validate the proposed approaches through simulations, demonstrating robust timing estimation and improved symbol error rate (SER) performance compared to existing time encoding receivers.



\end{itemize}

\section{System Architecture}
\label{section:sys_model}


We consider the pilot-aided communication model for a time encoding receiver illustrated in Fig. \ref{fig:IFTEM_sysmodel}. The transmit symbol sequence is divided into frames of length $L$. Each frame consists of a length-$L_{\mathrm{p}}$ pilot sequence $\{s_{l}^{(\mathrm{p})}\}_{l = 0}^{L_{\mathrm{p}}-1}$ and a length-$L_{\mathrm{d}}$ data sequence $\{s_{l}^{(\mathrm{d})}\}_{l = 0}^{L_{\mathrm{d}}-1}$, where $L_{\mathrm{d}} = L - L_{\mathrm{p}}$. The pilot sequence is a fixed sequence and known to both receiver and transmitter. Meanwhile, the symbols in the data sequence are drawn independently from an $M$-ary pulse amplitude modulation (PAM) constellation set $\mathcal{A}$.

The transmit symbol sequence is fed to a pulse shaping filter $p(t)$ in order to produce the waveform $X(t)$ that lies in the shift-invariant (SI) space. More precisely, the signal $X(t)$ can be written as
\begin{align*}
    X(t) =& \sum_{l=0}^{L_{\mathrm{p}}-1} s_{l}^{(\mathrm{p})}p(t-lT) + \sum_{l=0}^{L_{\mathrm{d}}-1} s_{l}^{(\mathrm{d})}p(t-(l+L_{\mathrm{p}})T),
\end{align*}
where $T$ is the symbol period. We assume that the pulse shaping filter $p(t)$ spans $2L_{\mathrm{f}}+1$ symbol periods, with $L_{\mathrm{f}}\in\mathbb{Z}$, and has negligible effect outside $t\in[-(L_{\mathrm{f}}+0.5)T,+(L_{\mathrm{f}}+0.5)T]$. We also assume that the guard interval of length $L_{\mathrm{f}}$ is placed between frames to eliminate inter-frame interference.

The channel introduces an unknown symbol timing offset $\tau_{\epsilon}\in[-0.5T,+0.5T)$ to the transmitted signal $X(t)$ and corrupts the delayed signal with additive white Gaussian noise (AWGN). The received signal has the form 
\begin{align}\label{eq:rx_sig}
    Y(t) =& X(t-\tau_{\epsilon}) + Z(t),
\end{align}
where $Z(t)$ is modeled as a zero-mean white Gaussian process with autocorrelation function $R_{Z}(\tau) = \frac{N_0}{2}\delta(\tau)$ and $N_0$ is the noise spectral density (in Watts/Hz).

At the time encoding receiver, the continuous-time received signal \( Y(t) \) is sampled using an IF-TEM characterized by parameters $b$, $\Delta$, and $\kappa$. The received signal is first shifted by a constant bias $b$, chosen to ensure that $\Pr(Y(t)+b < 0)$ remains negligible. The biased signal is then scaled by $\kappa$ and passed through an integrator, as shown in Fig.~\ref{fig:IFTEM_sysmodel}. The integrator output is continuously compared against a threshold $\Delta$. Each time the threshold is reached, a firing event occurs at time $t_k$, which generates a spike and triggers the integrator to reset. The \emph{firing times} $\{t_k\}$ obey the equilibrium condition
\begin{align}\label{eq:equilibrium_cond}
    \Delta = \frac{1}{\kappa}\int_{t_{k-1}}^{t_k}\big(Y(t)+b\big)\,dt.
\end{align}
Here, we assume that the integrator is at rest at some time $t = t_0$ known to the receiver and starts recording events from this time onward.

Because the firing times can grow without bound, it is often more practical to represent the signal by the intervals between successive firing times, known as \emph{time encodings}. These time encodings $T_{k} = t_{k} - t_{k-1}$ are bounded by
\begin{align}
    \frac{\kappa\Delta}{b+c_{\max}} \leq T_{k} \leq \frac{\kappa\Delta}{b-c_{\max}}\quad \forall k,
\end{align}
where $c_{\max}$ is the maximum amplitude of the IF-TEM input \cite{Lazar:2004}. Nonetheless, the firing times can be recovered from the time encodings using $t_{k} = t_0+\sum_{k'=1}^{k}T_{k'}$.

Suppose there are $K+1$ firing times. Since a length-$L$ frame contains pilot symbols and data symbols, we label the $K+1$ firing times for the data symbol periods and pilot symbol periods as $\{t_{k}^{(\mathrm{p})}\}_{k=0}^{K_{\mathrm{p}}}$ and $\{t_{k}^{(\mathrm{d})}\}_{k=1}^{K_{\mathrm{d}}}$, respectively, where $K = K_{\mathrm{p}}+K_{\mathrm{d}}$. Specifically, $t_{k}$ is tagged as a \emph{pilot firing time} $t_{k}^{(\mathrm{p})}$ if $t_k \in [-0.5T, (L_{\mathrm{p}}-0.5)T)$ within a frame and is tagged as a \emph{data firing time} $t_{k}^{(\mathrm{d})}$ otherwise. Since the pilot sequence appears before the data sequence, we set $t_{0}^{(\mathrm{d})} = {t_{K_{\mathrm{p}}}^{(p)}}$.


The recovery process at the time encoding receiver involves two phases: (1) the \emph{timing recovery phase} and (2) the \emph{data detection phase}. During the timing recovery phase, the time encoding receiver computes the maximum likelihood (ML) estimate of the channel-induced timing offset $\tau_{\epsilon}$ using the pilot firing times $\{t_{k}^{(\mathrm{p})}\}_{k=0}^{K_{\mathrm{p}}}$. After computing $\hat{\tau}_{\epsilon}^{(\mathrm{ML})}$, the receiver then  identifies the length-$L_{\mathrm{d}}$ data sequence sent by the transmitter from the data firing times $\{t_{k}^{(\mathrm{d})}\}_{k=0}^{K_{\mathrm{d}}}$. The detailed operation of the time encoding receiver during these phases is described in the following section.

\section{Time Encoding Receiver Algorithms}
\label{section:recovery}

\subsection{Log-likelihood for Joint ML Estimation}
A key element in the estimation process is the log-likelihood function, which captures the statistical dependence between the observed firing times and the unknown parameters of interest. In the following proposition, we derive the log-likelihood function for the time encoding-based communication model presented in Section \ref{section:sys_model}.
\begin{proposition}\label{proposition:log_likelihood}
     Consider the observed firing time vector $\mathbf{t} = [t_0,\cdots, t_K]^T$ and the known pilot sequence vector $\mathbf{s}_{\mathrm{p}} = [s_{0}^{(\mathrm{p})},\cdots, s_{L_{\mathrm{p}}-1}^{(\mathrm{p})}]^T$ available at the time encoding receiver. Let $\mathbf{s}_{\mathrm{d}} = [s_{0}^{(\mathrm{d})},\cdots, s_{L_{\mathrm{d}}-1}^{(\mathrm{d})}]^T$ denote the data sequence vector. The log-likelihood function for jointly estimating $\mathbf{s}_{\mathrm{d}}$ and $\tau_{\epsilon}$ is 
     \begin{align}\label{eq:loglikelihood_joint}
&\ln\mathcal{L}\left(\mathbf{s}_{\mathrm{d}},\tau_{\epsilon}|\mathbf{t},\mathbf{s}_{\mathrm{p}}\right)\nonumber \\
&\qquad=\; 
C_0(\mathbf{t}) - \left\|\mathbf{T}^{\frac{1}{2}}\left(\mathbf{y} - \mathbf{P}^{(\tau_{\epsilon})}\mathbf{s}_{\mathrm{p}}- \mathbf{G}^{(\tau_{\epsilon})}\mathbf{s}_{\mathrm{d}}\right)\right\|^2,
     \end{align}
     where the $k$-th element of $\mathbf{y}\in\mathbb{R}^{K\times 1}$ is
     \begin{align}
         y_k = \kappa\cdot\Delta - b\cdot [t_{k} - t_{k-1}],
     \end{align}
     the $k$-th row and $(l+1)$-th column of $\mathbf{P}^{(\tau_{\epsilon})}\in\mathbb{R}^{K\times L_{\mathrm{p}}}$ is
     \begin{align}
          P^{(\tau_{\epsilon})}_{k,l+1} = \int_{t_{k-1}}^{t_k}p(t-lT-\tau_{\epsilon})dt,
     \end{align}
      the $k$-th row and $(l+1)$-th column of $\mathbf{G}^{(\tau_{\epsilon})}\in\mathbb{R}^{K\times L_{\mathrm{d}}}$ is
      \begin{align}
          G^{(\tau_{\epsilon})}_{k,l+1} = \int_{t_{k-1}}^{t_k}p(t-(L_{\mathrm{p}}+l)T-\tau_{\epsilon})dt,
     \end{align}
     $\mathbf{T}\in\mathbb{R}^{K\times K}$ is a diagonal matrix whose $k$-th diagonal entry is $[\mathbf{T}]_{k,k} = \frac{1}{t_{k}-t_{k-1}}$, and $C_{0}(\mathbf{t})$ is a term independent of the unknown parameters.
\end{proposition}
\begin{proof}
See Appendix \ref{proof:appendix_A}.
\end{proof}
A joint maximum likelihood (ML) estimation of $\left(\mathbf{s}_{\mathrm{d}},\tau_{\epsilon}\right)$  can therefore be formulated as
\begin{align}
\left(\mathbf{\hat{s}}_{\mathrm{d}}^{(\mathrm{ML})},\hat{\tau}_{\epsilon}^{(\mathrm{ML})}\right) = \underset{\substack{\mathbf{s}_{\mathrm{d}}\in\mathcal{A}^{L_{\mathrm{d}}}\\\tau_{\epsilon}\in\left[-\frac{T}{2},+\frac{T}{2}\right]}}{\arg\max}\ln\mathcal{L}\left(\mathbf{s}_{\mathrm{d}},\tau_{\epsilon}|\mathbf{t},\mathbf{s}_{\mathrm{p}}\right),
\end{align}
which constitutes a non-convex and NP-hard optimization problem. To make the estimation process computationally tractable, we decompose it into two stages: (a) a timing recovery phase, where the timing offset $\tau_{\epsilon}$ is estimated from the firing times using the pilot sequence, and (b) a data detection phase, where the data symbols are recovered using the estimated timing offset and the observed data firing times.

\begin{algorithm}[t]\label{algo:ML_timing_opt}
\DontPrintSemicolon
  \KwInput{$N_{\mathrm{guess}}$, $\{t_{k}^{(\mathrm{p})}\}_{k=0}^{\tilde{K}_{\mathrm{p}}}$,$\{s_l^{(\mathrm{p})}\}_{l=0}^{L_{\mathrm{p}}-1}$} 
  \KwOutput{$\hat{\tau}_{\epsilon}^{(\mathrm{ML})}$ }
  $\hat{\tau}_{\epsilon}^{(\mathrm{ML})} = \mathrm{NULL}$  \tcp*{Initialize}
  \tcc{Define the objective function in \eqref{eq:ML_timing}}
  $\mathrm{ML\_Obj(x)} = \sum_{k = 1}^{\tilde{K}_{\mathrm{p}}} \frac{\left(\tilde{y}_{k}^{(\mathrm{p})} - \sum_{l = 0}^{L_{\mathrm{p}}-1}s_l^{\mathrm{(p)}}\tilde{P}^{(\mathrm{x})}_{k,l+1}\right)^2}{2(t_{k}^{(\mathrm{p})} - t_{k-1}^{(\mathrm{p})})}$\;
   \For{$\ell = 0\;\mathbf{ to }\;N_{\mathrm{guess}}-1$} 
   {
        $\tau_{\mathrm{init}}^{(\ell)} = -0.5 + \frac{\ell}{N_{\mathrm{guess}} - 1}$\tcp*{$\ell$-th initial guess}
        \tcc{Get root of \eqref{eq:first_order_opt} with initial guess $\tau_{\mathrm{init}}^{(\ell)}$}
        $\tau_{\mathrm{candidate}}^{(\ell)}$ = result of applying Newton's method to  \eqref{eq:first_order_opt} with $\tau_{\mathrm{init}}^{(\ell)}$ as initial guess.\;
        \If{$\mathrm{ML\_Obj(\hat{\tau}_{\epsilon}^{(\mathrm{ML})})} > \mathrm{ML\_Obj(\hat{\tau}_{\mathrm{candidate}}^{(\ell)})}$} {
            $\hat{\tau}_{\epsilon}^{(\mathrm{ML})}  = \hat{\tau}_{\mathrm{candidate}}^{(\ell)}$  \tcp*{update the estimate}
        }
   }
\caption{ML Timing Offset Estimation}
\end{algorithm}

\subsection{Timing Recovery Phase} \label{subsection:timing_recovery}
Due to the absence of a fixed-rate clock in an asynchronous sampling system, it is not straightforward to identify the appropriate timing boundaries of a symbol period given the firing times. To address this, the receiver estimates the timing offset $\tau_{\epsilon}$ from the pilot firing times and the known pilot symbol sequence. This subsection presents the algorithm for obtaining the ML estimate of $\tau_{\epsilon}$.

\begin{figure*}[!ht]
\begin{align}\label{eq:first_order_opt}
\frac{\ln\mathcal{L}\left(\tau_{\epsilon}|\mathbf{\tilde{t}}_{\mathrm{p}},\mathbf{s}_{\mathrm{p}}\right)}{\partial \tau_{\epsilon}}\Bigg|_{\tau_{\epsilon} = \hat{\tau}_{\epsilon}^{(\mathrm{ML})}} &= 
     \sum_{k = 1}^{\tilde{K}_{\mathrm{p}}}\left( \frac{\tilde{y}_{k}^{(\mathrm{p})} - \sum_{l = 0}^{L_{\mathrm{p}}-1}s_l^{\mathrm{(p)}}\tilde{P}^{\left({\tau}_{\epsilon}\right)}_{k,l+1}}{T_{k}^{(\mathrm{p})}}\right)\left(\sum_{l = 0}^{L_{\mathrm{p}}-1}s_{l}^{(\mathrm{p})}\cdot\frac{\partial \tilde{P}^{\left({\tau}_{\epsilon}\right)}_{k,l+1}}{\partial \tau_{\epsilon}}\right)\Bigg|_{\tau_{\epsilon} = \hat{\tau}_{\epsilon}^{(\mathrm{ML})}} = 0.
         \end{align}
         \hrulefill
\end{figure*} 

To mitigate the effect of ISI from unknown data symbols, only the first $\tilde{L}_{\mathrm{p}} = L_{\mathrm{p}} - L_{\mathrm{f}}$ pilot symbol periods are considered in the timing recovery phase. We will call $\tilde{L}_{\mathrm{p}}$ the \emph{effective pilot length}. Let $\boldsymbol{\tilde{t}}_{\mathrm{p}} = [t_{0}^{(\mathrm{p})},\cdots, t_{\tilde{K}_{\mathrm{p}}}^{(\mathrm{p})}]^T$ denote the vector of $\tilde{K}_{\mathrm{p}}+1$ pilot firing times within the interval $[-0.5T, (\tilde{L}_{\mathrm{p}} - 0.5)T)$. By using this subset of pilot firing times, the log-likelihood function in \eqref{eq:loglikelihood_joint} reduces to
\begin{align}\label{eq:loglikelihood_timing}
\ln\mathcal{L}\left(\tau_{\epsilon}|\mathbf{\tilde{t}}_{\mathrm{p}},\mathbf{s}_{\mathrm{p}}\right) = C_{1}(\mathbf{\tilde{\mathbf{t}}}_{\mathrm{p}}) - \|\mathbf{\tilde{T}}^{\frac{1}{2}}_{\mathrm{p}}\left(\mathbf{\tilde{y}}_{\mathrm{p}} - \mathbf{\tilde{P}}^{(\tau_{\epsilon})}\mathbf{s}_{\mathrm{p}}\right)\|^2,
\end{align}
where the $k$-th element of $\mathbf{\tilde{y}}_{\mathrm{p}}\in\mathbb{R}^{\tilde{K}_{\mathrm{p}}\times 1}$ is   
\begin{align}
    \tilde{y}_k^{(\mathrm{p})} = \kappa\cdot\Delta - b\cdot [t_{k}^{(\mathrm{p})} - t_{k-1}^{(\mathrm{p})}],
\end{align}
the $k$-th row and $(l+1)$-th column of $\mathbf{\tilde{P}}^{(\tau_{\epsilon})}\in\mathbb{R}^{\tilde{K}_{\mathrm{p}}\times L_{\mathrm{p}}}$ is
     \begin{align}
          \tilde{P}^{(\tau_{\epsilon})}_{k,l+1} = \int_{t_{k-1}^{(\mathrm{p})}}^{t_k^{\mathrm{(p)}}}p(t-lT-\tau_{\epsilon})dt,
     \end{align}
and $\mathbf{\tilde{T}_{\mathrm{p}}}\in\mathbb{R}^{\tilde{K}_{\mathrm{p}}\times \tilde{K}_{\mathrm{p}}}$ is a diagonal matrix whose $k$-th diagonal element is $[\mathbf{\tilde{T}}_{\mathrm{p}}]_{k,k}=\frac{1}{t_{k}^{(\mathrm{p})} - t_{k-1}^{(\mathrm{p})} }$.

The ML estimate for $\tau_{\epsilon}$ can be formulated as
\begin{align}\label{eq:ML_timing}
\hat{\tau}_{\epsilon}^{(\mathrm{ML})} = \underset{\tau_{\epsilon}\in[-0.5T,0.5T]}{\arg\min} \;\sum_{k = 1}^{\tilde{K}_{\mathrm{p}}} \frac{\left(\tilde{y}_{k}^{(\mathrm{p})} - \sum_{l = 0}^{L_{\mathrm{p}}-1}s_l^{\mathrm{(p)}}\tilde{P}^{(\tau_{\epsilon})}_{k,l+1}\right)^2}{2(t_{k}^{(\mathrm{p})} - t_{k-1}^{(\mathrm{p})})},
 \end{align}
the objective function in \eqref{eq:ML_timing} can be interpreted as a weighted sum of squared errors, where each residual is weighted inversely by the corresponding inter-spike interval duration.

A necessary condition for \(\hat{\tau}_{\epsilon}^{(\mathrm{ML})}\) is the first-order optimality condition given in \eqref{eq:first_order_opt}. Since the objective in \eqref{eq:ML_timing} is non-convex with respect to $\tau_{\epsilon}$, multiple stationary points may exist. We therefore employ Newton’s method with several initial guesses on \eqref{eq:first_order_opt} to locate potential solutions and select the one yielding the minimum objective value. The complete procedure is summarized in Algorithm~\ref{algo:ML_timing_opt}.

\subsection{Data Detection Phase}\label{subsection:data_detection}

After estimating the timing offset $\tau_{\epsilon}$, our next step is to recover the data sequence vector $\mathbf{s}_{\mathrm{d}}$ from the $K_{\mathrm{d}}+1$ data firing times $\mathbf{t}_{\mathrm{d}} = [t_0^{(\mathrm{d})},\cdots,t_{K_{\mathrm{d}}}^{(\mathrm{d})}]^T$. Previous approaches \cite{Bernardo2025SIPS,Katsaros:2025} detect data symbols by counting the number of firing events per symbol period and mapping this count to the corresponding symbol in the constellation. While this technique yields a coarse estimate of symbol amplitudes, it ignores the finer temporal variations within each symbol interval. To improve detection accuracy, we use $\mathbf{t}_{\mathrm{d}}$ directly instead of the number of firing times per symbol period to generate the symbol sequence estimate.

By substituting the ML timing estimate $\hat{\tau}_{\epsilon}^{(\mathrm{ML})}$ into the likelihood function in \eqref{eq:loglikelihood_joint} and omitting constant terms that do not depend on $\mathbf{s}_{\mathrm{d}}$, the ML estimation of $\mathbf{s}_{\mathrm{d}}$ becomes
\begin{align}\label{eq:ML_detection}
\mathbf{\hat{s}}_{\mathrm{d}}^{(\mathrm{ML})} 
&= \underset{\mathbf{s}_{\mathrm{d}}\in\mathcal{A}^{L_{\mathrm{d}}}}{\arg\min}\;\left\|\mathbf{T}_{\mathrm{d}}^{\frac{1}{2}}\left(\mathbf{y}_{\mathrm{d}}- \mathbf{\bar{G}}\mathbf{s}_{\mathrm{d}}\right)\right\|^2,
\end{align}
where the $k$-th element of $\mathbf{y}_{\mathrm{d}}\in\mathbb{R}^{K_{\mathrm{d}}\times 1}$ is given by
\begin{align}\label{eq:y_k_d}
         y_k^{(\mathrm{d})} =& \kappa\cdot\Delta - b\cdot [t_{k}^{(\mathrm{d})} - t_{k-1}^{(\mathrm{d})}] \nonumber\\
         &- \sum_{l=L_{\mathrm{p}}-1-L_{\mathrm{f}}}^{L_{\mathrm{p}}-1}s_{l}^{(\mathrm{p})}\int_{t_{k-1}^{(\mathrm{d})}}^{t_k^{(\mathrm{d})}}p\left(t-lT-\hat{\tau}_{\epsilon}^{(\mathrm{ML})}\right)dt,
\end{align}
the $k$-th row and $(l+1)$-th column of $\mathbf{\bar{G}}\in\mathbb{R}^{K_{\mathrm{d}}\times L_{\mathrm{d}}}$ is
      \begin{align}
          \bar{G}^{(\hat{\tau}_{\epsilon}^{(\mathrm{ML})})}_{k,l+1} = \int_{t_{k-1}^{(\mathrm{d})}}^{t_k^{(\mathrm{d})}}p\left(t-(L_{\mathrm{p}}+l)T-\hat{\tau}_{\epsilon}^{(\mathrm{ML})}\right)dt,
     \end{align}
and $\mathbf{T}_{\mathrm{d}}\in\mathbb{R}^{K_{\mathrm{d}}\times K_{\mathrm{d}}}$ is a diagonal matrix whose $k$-th diagonal entry is $
[\mathbf{T}_{\mathrm{d}}]_{k,k} = \frac{1}{t_{k}^{(\mathrm{d})} - t_{k-1}^{(\mathrm{d})}}$.
The final term in \eqref{eq:y_k_d} accounts for interference contributions from the known pilot sequence.

A variety of linear and nonlinear sequence detection techniques can be used to approximate the solution of the ML detection problem in \eqref{eq:ML_detection} \cite{Larsson:2009}. In this work, we adopt a suboptimal yet computationally efficient linear detector called the ZF receiver. The ZF pre-estimate is given by
\begin{align}
\mathbf{\tilde{s}}_{\mathrm{d}}^{(\mathrm{ZF})} = \left[\mathbf{\bar{G}}^T\mathbf{T}_{\mathrm{d}}\mathbf{\bar{G}}\right]^{-1}\mathbf{\bar{G}}^T\mathbf{T}_{\mathrm{d}}\mathbf{y}_{\mathrm{d}}.
\end{align}
Subsequently, the elements of the ZF pre-estimate vector $\mathbf{\tilde{s}}_{\mathrm{d}}^{(\mathrm{ZF})}$ are mapped to the nearest symbol in the $M$-PAM constellation through hard decision decoding. The result of the hard decision decoding yields the length-$L_{\mathrm{d}}$ sequence estimate $\mathbf{\hat{s}}_{\mathrm{d}} =~ [\hat{s}_{0}^{(\mathrm{d})},\cdots,\hat{s}_{L_{\mathrm{d}}-1}^{(\mathrm{d})}]^T$.

It is important to note that the ZF data sequence estimate derived in this work fundamentally differs from the ZF detector presented in \cite{Bernardo2025SIPS}. In our formulation, the ZF receiver operates on the data firing time vector $\mathbf{t}_{\mathrm{d}}$, which retains the fine temporal structure of the firing times. In contrast, the approach in \cite{Bernardo2025SIPS} relies on the vector $\mathbf{N} = [N_0, \cdots, N_{L_{\mathrm{d}}-1}]^T$, where each $N_l$ denotes the total number of firing events within the interval $[(L_{\mathrm{p}}+l - 0.5)T, (L_{\mathrm{p}}+l + 0.5)T)$. By leveraging the precise timing information rather than aggregate spike counts, the proposed receiver achieves improved detection accuracy, particularly in low firing rate regimes.

\section{Numerical Results}
\label{section:numerical_results}

\begin{figure}[t!]
    \hspace{-.5cm}
    \includegraphics[scale = .5]{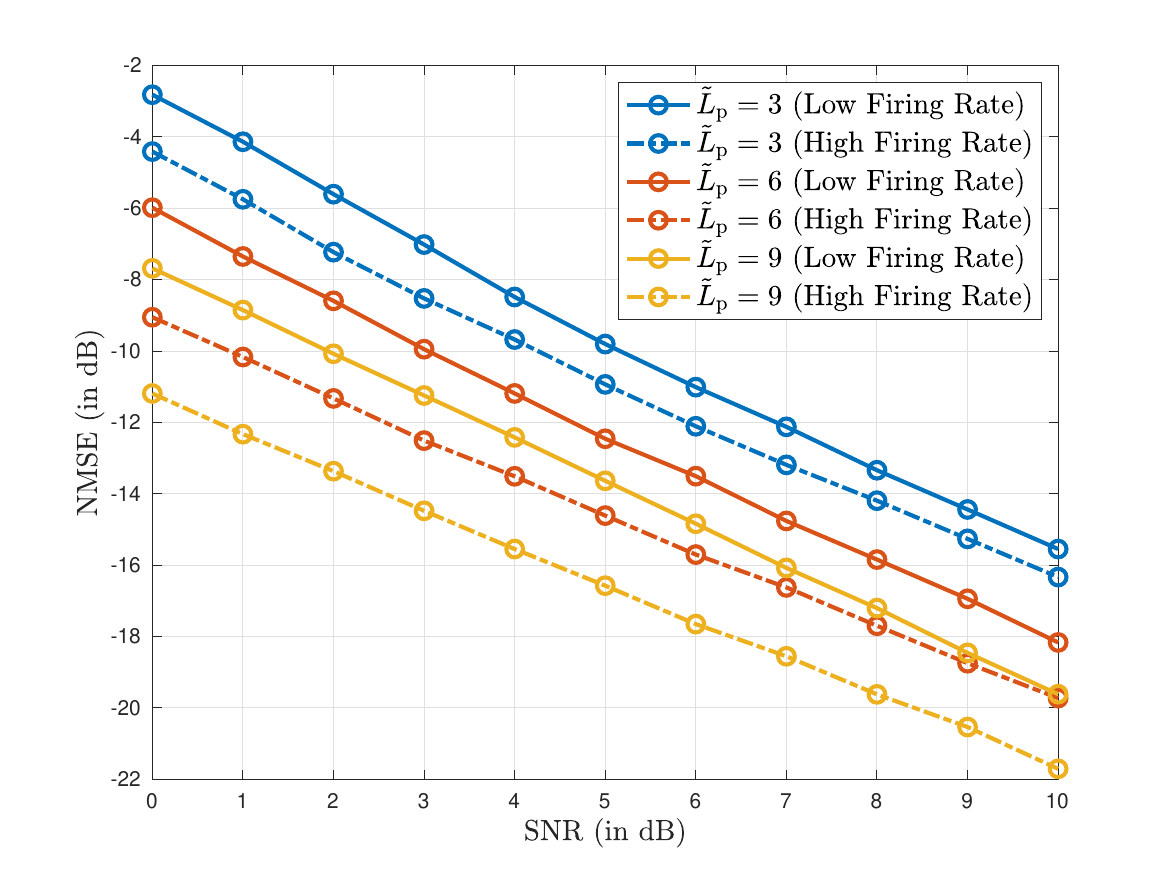}
    \caption{NMSE (in dB) vs. SNR (in dB) of the proposed timing recovery for the integrate-and-fire time encoding receiver under different operating modes ($b = 1.50$ and $b = 4.50$) and different effective pilot lengths ($\tilde{L}_{\mathrm{p}} = 3$, $\tilde{L}_{\mathrm{p}} = 6$, and $\tilde{L}_{\mathrm{p}} = 9$).}
    \label{fig:NMSE_curves}
\end{figure}

In this section, we evaluate the performance of the proposed timing recovery and data detection scheme for integrate-and-fire time encoding receivers under various system parameters. We consider a pilot-aided transmission model with a total frame length of $L=100$ symbols and a $4$-PAM modulation scheme with constellation $\mathcal{A} = \{-3,-1,+1,+3\}$. The length-$L_{\mathrm{p}}$ pilot sequence is composed of alternating symbols of $+1$ and $-1$, which facilitates timing acquisition. We also consider a Gaussian pulse shaping filter whose impulse response is 
\begin{align}
    p_{\mathrm{Gauss}}(t) = \frac{\sqrt{\pi}}{a}\exp\left(-\frac{\pi^2t^2}{a^2}\right),
\end{align}
where 
$a$ serves as a shaping parameter that is inversely related to the product of the 3-dB bandwidth $B_{\mathrm{3dB}}$ and the symbol period $T$. More precisely, the parameter $a$ is given by 
\begin{align}
    a = \frac{1}{B_{3\mathrm{dB}}\cdot T}\sqrt{\frac{\ln2}{2}}.
\end{align}
For the experiments, we set $T = 1$ second and $B_{\mathrm{3dB}} = 0.5$~Hz. The filter span of $p(t)$ is set to 5 symbol periods ($L_{\mathrm{f}} = 2$). 

We implement the IF-TEM component of the integrate-and-fire time encoding receiver using the time encoding and decoding toolkit \cite{bionet_ted_matlab} from the Bionet Group of Columbia University. The firing threshold and the integrator constant are set to $\Delta = 1.0$ and $\kappa = 0.1$, respectively. Two operating modes are considered: (a) a \emph{low firing rate mode} with bias $b = 1.50$ and (b) a \emph{high firing rate mode} with bias $b = 4.50$. Using the transmitted signal $X(t)$ described above and assuming a noiseless environment, the average spike rates for the low and high firing rate modes are approximately 15.34 and 44.91 firing events per second, respectively.

\begin{figure}[t!]
    \hspace{-.5cm}
    \includegraphics[scale = .5]{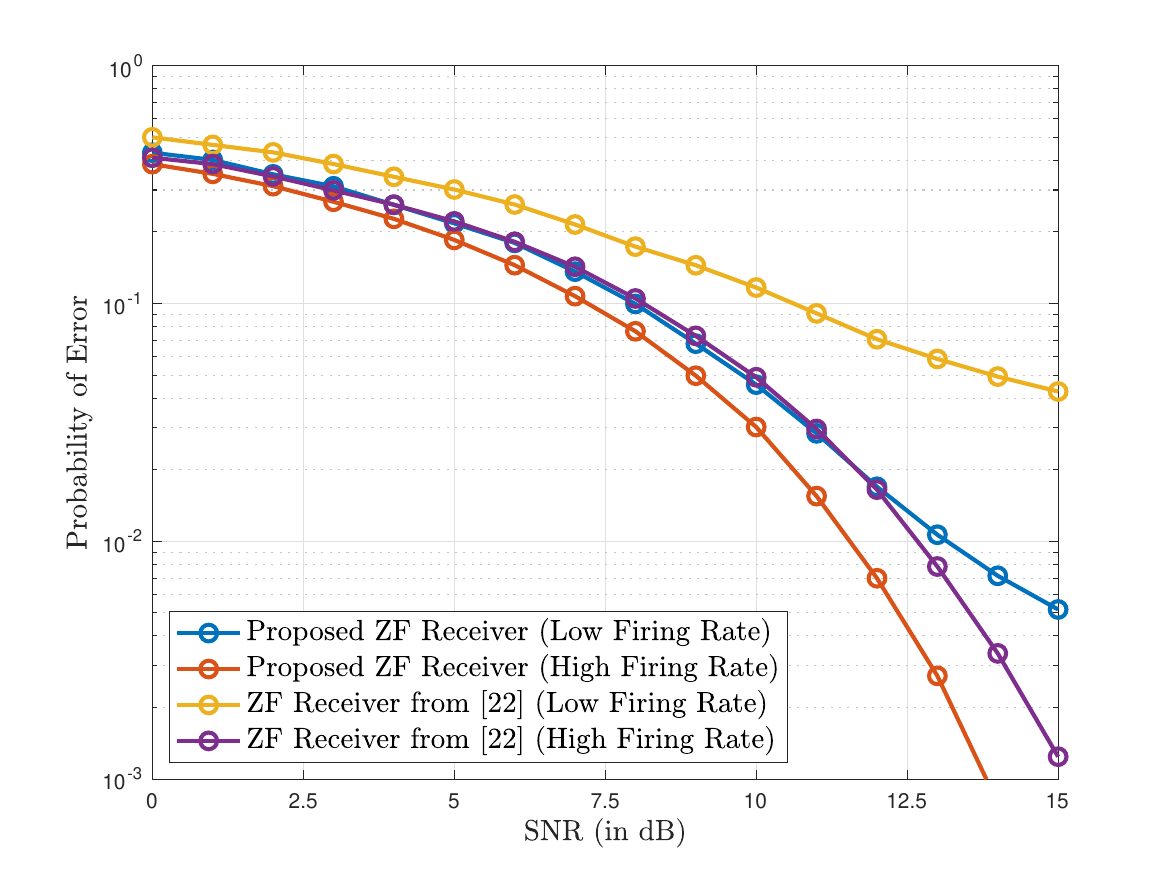}
    \caption{Symbol Error Rate vs. SNR (in dB) of the proposed ZF-based sequence detection for the integrate-and-fire time encoding receiver under different operating modes ($b = 1.50$ and $b = 4.50$). Superimposed in the plot is the ZF receiver presented in \cite{Bernardo2025SIPS}.}
    \label{fig:SEP_curves}
\end{figure}

We first evaluate the accuracy of the proposed timing recovery algorithm based on the ML formulation in \eqref{eq:ML_timing}. We assume that the symbol timing offset is drawn uniformly from the interval $[-0.5T,+0.5T]$, and the performance metric considered is the normalized mean squared error (NMSE) between the ML estimate and true timing offsets. The number of initial guesses in Algorithm \ref{algo:ML_timing_opt} is set to $N_{\mathrm{guess}} = 5$. Fig.~\ref{fig:NMSE_curves} depicts the NMSE as a function of signal-to-noise ratio (SNR) in dB and under different effective pilot lengths $\tilde{L}_{\mathrm{p}}$. It can be observed that the timing recovery performance improves as SNR increases, reflecting the fact that the observed firing times become more reliable at higher SNR levels. Additionally, longer effective pilot length leads to lower NMSE, indicating that more pilot firing times provide richer temporal information that is relevant for estimating the unknown timing offset. 

The effect of the firing rate is also evident: in the low firing rate mode ($b = 1.50$), the NMSE remains slightly higher at low SNRs due to the sparser sampling of the input signal, whereas in the high firing rate mode ($b = 4.50$), the denser spike occurrences allow for more precise timing estimates across all SNR levels. The gap between the NMSE of the low and high firing rate modes also increases with the pilot length, which can be attributed to the growing difference in the number of pilot firing times between the two modes as the effective pilot length increases.

After timing recovery, the data sequence $\mathbf{s}_{\mathrm{d}}$ is estimated from the observed data firing times $\mathbf{t}_{\mathrm{d}}$ using the ZF-based detection approach described in Section \ref{subsection:data_detection}. Fig.~\ref{fig:SEP_curves} compares the SER of our proposed ZF receiver for integrate-and-fire time encoding machine against the ZF receiver presented in \cite{Bernardo2025SIPS}. The results clearly show that the proposed sequence detection method outperforms the spike count-based ZF detector \cite{Bernardo2025SIPS} for both operating modes. In fact, our proposed ZF receiver operating in the low firing rate mode has performance close to the spike count-based ZF receiver operating in the high firing rate mode from SNR = 0 dB to SNR = 12 dB. This can be attributed to the finer temporal resolution exploited by our proposed approach.

\section{Conclusion}
\label{section:conclussion}

This paper presented a statistical framework for timing recovery and data detection in integrate-and-fire time encoding receivers. By modeling the relationship between firing times, timing offset, and transmitted symbols, we formulated a maximum likelihood estimation problem that provides a unified treatment of symbol timing estimation and sequence detection tasks. A practical two-stage receiver architecture was then developed, consisting of an ML-based timing recovery followed by a zero-forcing-based sequence detection. Simulation results demonstrated that the proposed timing recovery algorithm achieves accurate symbol timing estimates with reasonably short pilot sequences, while the proposed sequence detection method significantly improves SER performance compared to existing spike count-based receivers.



\section*{Acknowledgements}\label{section:acknowledgement}

The author acknowledges the DOST-Science Education Institute (DOST-SEI) and the Engineering Research and Development for Technology (ERDT) Program for providing financial support for conference attendance, and the Office of the Chancellor of the University of the Philippines Diliman, through the Office of the Vice Chancellor for Research and Development, for funding support through the PhD Incentive Award Grant 252509 YEAR 1.




\begin{appendices}
\section{Proof of Proposition \ref{proposition:log_likelihood}}\label{proof:appendix_A}

From \eqref{eq:equilibrium_cond}, the firing events satisfy
\begingroup
\allowdisplaybreaks
    \begin{align*}
        \frac{\kappa\cdot \Delta}{T_{k}} =& \frac{1}{T_{k}}\int_{t_{k-1}}^{t_k}[Y(t) + b]dt\nonumber\\
        = & \sum_{l = 0}^{L_{\mathrm{d}}-1}s_l^{(\mathrm{d})}\frac{\int_{t_{k-1}}^{t_k}p(t - (l+L_{\mathrm{p}})T-\tau_\epsilon)dt}{T_{k}} + b \\
        &+\sum_{l = 0}^{L_{\mathrm{p}}-1}s_l^{(\mathrm{p})}\frac{\int_{t_{k-1}}^{t_k}p(t - lT-\tau_\epsilon)dt}{T_{k}} + \underbrace{\frac{\int_{t_{k-1}}^{t_k}Z(t)dt}{T_{k}}}_{Z_{k}}\nonumber
    \end{align*}
\endgroup
 where $Z_k\sim\mathcal{N}\left(0,\frac{N_0}{2\cdot T_{k}}\right)$. The second line follows from the shift-invariant (SI) structure of $X(t)$. The last term is Gaussian due to the property of continuous time Gaussian processes.
 
 Since $\mathbf{Z} = [Z_1,\cdots, Z_K]^T \sim \mathcal{N}\left(\mathbf{0},\frac{N_0}{2}\mathbf{T}\right)$, the conditional probability density function (PDF) of $\left({\mathbf{s}^{(\mathrm{d})}},\tau_{\epsilon}\right)$ given the firing times $\mathbf{t}$ and pilot sequence $\mathbf{s}^{(\mathrm{p})}$ can be written as
 \begin{align*}
f\left({\mathbf{s}^{(\mathrm{d})}},\tau_{\epsilon}|\mathbf{t},\mathbf{s}^{(\mathrm{p})}\right) = \frac{\exp\left(-\frac{\left\|\mathbf{T}^{\frac{1}{2}}\left(\mathbf{y} - \mathbf{P}^{(\tau_{\epsilon})}\mathbf{s}^{(\mathrm{p})}- \mathbf{G}^{(\tau_{\epsilon})}\mathbf{s}^{(\mathrm{d})}\right)\right\|^2}{N_0}\right)}{\left(\pi\cdot N_0\right)^{\frac{K}{2}} \left(\prod_{k=1}^{K} T_{k}^{-1/2}\right)},
 \end{align*}
 where $\mathbf{y}$, $\mathbf{P}^{(\tau_{\epsilon})}$, $\mathbf{G}^{(\tau_{\epsilon})}$, and $\mathbf{T}$ are defined in Proposition \ref{proposition:log_likelihood}. The proof is completed by taking the natural logarithm of the conditional PDF and collecting all terms that do not depend on $\mathbf{s}^{(\mathrm{d})}$ and $\tau_{\epsilon}$ in $C_0(\mathbf{t})$.
 

\end{appendices}



\renewcommand*{\bibfont}{\footnotesize}
\begingroup
\footnotesize  
\printbibliography
\endgroup






\end{document}